\documentclass[aps,twocolumn,prl,nofootinbib,longbibliography,superscriptaddress]{revtex4-2}

\usepackage{amssymb,amsfonts}
\usepackage{amsmath}
\usepackage{mathrsfs}
\usepackage{tikz}
\usepackage{amsthm}
\usetikzlibrary{decorations.pathreplacing}

\usepackage[colorlinks, allcolors=blue!70!black, linktocpage]{hyperref}
\usepackage{accents}
\usepackage{graphicx}
\usepackage[utf8]{inputenc}
\usepackage{float}

\usepackage{titlesec}
\usepackage{cases}
\usepackage{mathtools}
\usepackage{color}
\usepackage{cleveref}
\usepackage{braket}
\usepackage{comment}
\usepackage{extarrows}
\usepackage{booktabs}

\definecolor{Gray}{gray}{0.95}
\definecolor{LightCyan}{rgb}{0.88,1,1}

\def\Tr{{\text{Tr}}}

\newcommand{\pbn}{p_{\mathrm{bn}}}

\crefname{lem}{lemma}{lemmas}
\crefname{thm}{theorem}{theorems}
\crefname{cor}{corollary}{corollaries}
\crefname{rem}{remark}{remarks}
\crefname{prop}{proposition}{propositions}
\theoremstyle{plain}
\newtheorem{theorem}{Theorem}
\newtheorem{proposition}{Proposition}
\newtheorem{lemma}{Lemma}
\newtheorem{corollary}{Corollary}
\theoremstyle{definition}
\newtheorem{remark}{Remark}
\theoremstyle{plain}
\newcommand{\cyc}{\operatorname{cyc}}
\newcommand{\dA}{d_A}
\newcommand{\dB}{d_B}
\newcommand{\dC}{d_C}
\newcommand{\rhoTB}{\rho_{AB}^{T_B}}
\newcommand{\bineg}{{|\rho_{AB}^{T_B}|^{T_B}}}
\newcommand{\Chalf}{C_{1/2}}
\newcommand{\ZZ}{\mathbb{Z}}
\newcommand{\Cat}{C}

\newcommand{\specplot}[7]{%
  \begin{scope}[shift={(#1,#2)}]
    \node[font=\scriptsize, anchor=south] at (0.85,1.06) {#3};
    \draw[->,gray!75,line width=0.4pt] (-0.05,0) -- (1.72,0)
        node[anchor=west,font=\scriptsize]{$#4$};
    \draw[->,gray!75,line width=0.4pt] (#5,-0.05) -- (#5,1.0);
    \filldraw[draw=blue!55!black,fill=blue!15,line width=0.7pt]
        plot[smooth,tension=0.55] coordinates {#6} -- cycle;
    \draw[blue!55!black,line width=1.1pt] (#5,0) -- (#5,#7); 
  \end{scope}}

\begin{document}
\count\footins = 1000

\title{Logarithmic negativity typically equals exact entanglement cost}

\author{Bowen Ouyang}
\email{bowenouyang@brandeis.edu}
\affiliation{SJTU Paris Elite Institute of Technology, Shanghai Jiao Tong University,
Shanghai 200240, China}
\affiliation{Martin Fisher School of Physics, Brandeis University, Waltham, Massachusetts 02453, USA}
\author{Jonah Kudler-Flam}
\email{jkudlerflam@ias.edu}
\affiliation{School of Natural Sciences, Institute for Advanced Study, Princeton, NJ 08540, USA}

\begin{abstract}
Quantum entanglement plays a leading role in modern understanding of physical systems, from quantum phases of matter to quantum gravity. In quantum information theory, one seeks operationally meaningful quantifiers of entanglement, which for large quantum systems are notoriously difficult to evaluate due to the lack of computationally efficient algorithms. In this Letter, we show that for large random induced mixed states the logarithmic negativity, an efficiently computable entanglement measure, generically coincides with the exact entanglement cost under positive-partial-transpose-preserving operations, thereby acquiring a precise operational interpretation. Our results establish logarithmic negativity as an exact characterization of entanglement in generic many-body states and provide a tractable route for quantifying entanglement in complex quantum systems.
\end{abstract}

\maketitle

\textit{Introduction.}---
Entanglement is the defining resource of quantum physics, organizing our understanding of systems as varied as quantum phases of matter and quantum gravity. For pure states, entanglement entropy is characterized by a single quantity, the von Neumann entropy, which determines all asymptotic conversion rates under local operations and classical communication (LOCC)~\cite{BBPS1996}. For mixed states that arise ubiquitously, whenever a subsystem is considered or systems at finite temperature, the situation is far richer. A host of inequivalent measures exist, and the operationally meaningful ones are the most sought after. Chief among these are the distillable entanglement, the asymptotic rate at which Bell pairs can be extracted by LOCC, and the entanglement cost, the rate at which Bell pairs must be consumed to prepare the state \cite{BennettMixed1996,HaydenHorodeckiTerhal2001}. Computing these quantities is one of the central goals of quantum information theory. Unfortunately, both are defined by intractable asymptotic optimizations, rendering them essentially uncomputable for large, many-body quantum systems of physical interest.

This has motivated the search for entanglement measures that are simultaneously meaningful and efficient to evaluate. The most successful candidate is built from the Peres-Horodecki positive partial transpose (PPT) criterion \cite{Peres1996,HorodeckiSep1996}. A separable (unentangled) state remains positive under partial transposition, so any negative eigenvalue of the partial transpose witnesses entanglement. Logarithmic negativity, an entanglement monotone computed from a single matrix diagonalization, quantifies this entanglement \cite{ZHSL1998,VidalWerner2002,Plenio2005}. Its computability came at the price of a precise operational interpretation. In this Letter, we close this gap by showing that the negativity generically equals the exact PPT entanglement cost, and therefore acquires a sharp operational meaning for almost every state.\footnote{A related measure based on the PPT criterion is the $\kappa$-entanglement~\cite{WangWilde2020}. This was introduced because it is closely related to the exact PPT entanglement cost yet is efficiently computable via semi-definite programming. Our result also fixes the $\kappa$-entanglement, and the hierarchy of entanglement measures in \cite{LamiMeleRegula2025}, which satisfies $\mathcal{E} \leq \kappa \leq E_{ppt}$ for all states \cite{WangWilde2023,LamiMeleRegula2025}. Since $E_{ppt} = \mathcal{E}$ (exactly or to leading order) for typical states, the $\kappa$-entanglement is likewise generically equal to the logarithmic negativity, and hence to the exact entanglement cost.}

Given a density matrix $\rho_{AB}$ on a bipartite Hilbert space, $\mathcal{H}_A\otimes \mathcal{H}_B$ with dimension $d_Ad_B$, the logarithmic negativity (henceforth referred to as the negativity) is defined as
\begin{align}
    \mathcal{E}(\rho_{AB}) = \log \| \rho_{AB}^{T_B}\|_1,
\end{align}
where $(\cdot)^{T_B}$ is the partial transpose and $\| \cdot \|_1$ is the trace norm. Throughout, $\log$ denotes the natural logarithm.

The negativity bounds the distillable entanglement from above \cite{VidalWerner2002,Rains1999}. It also lower bounds the exact PPT entanglement cost, $E_{ppt}$, which is the number of Bell pairs that must be consumed to prepare $\rho_{AB}$ using PPT preserving operations. Indeed, $E_{ppt}$ has been proven to be bounded as \cite{APE2003}
\begin{align}
\begin{aligned}
      \mathcal{E}(\rho_{AB}) &\leq E_{ppt}(\rho_{AB}) \leq \log Z(\rho_{AB})
    \\
    Z(\rho_{AB})&\equiv \| \rho_{AB}^{T_B}\|_1  + d_Ad_B\max(0,-\xi_{\min}(|\rho_{AB}^{T_B}|^{T_B}))
    \label{eq:bounds}
\end{aligned}
\end{align}
where $\xi_{\min}(|\rho_{AB}^{T_B}|^{T_B})$ is the minimum eigenvalue of the binegativity operator. If the binegativity $|\rho_{AB}^{T_B}|^{T_B}$ is positive, as it is e.g.~for all two-qubit states \cite{Ishizaka2004}, we see that the PPT entanglement cost is exactly the negativity.

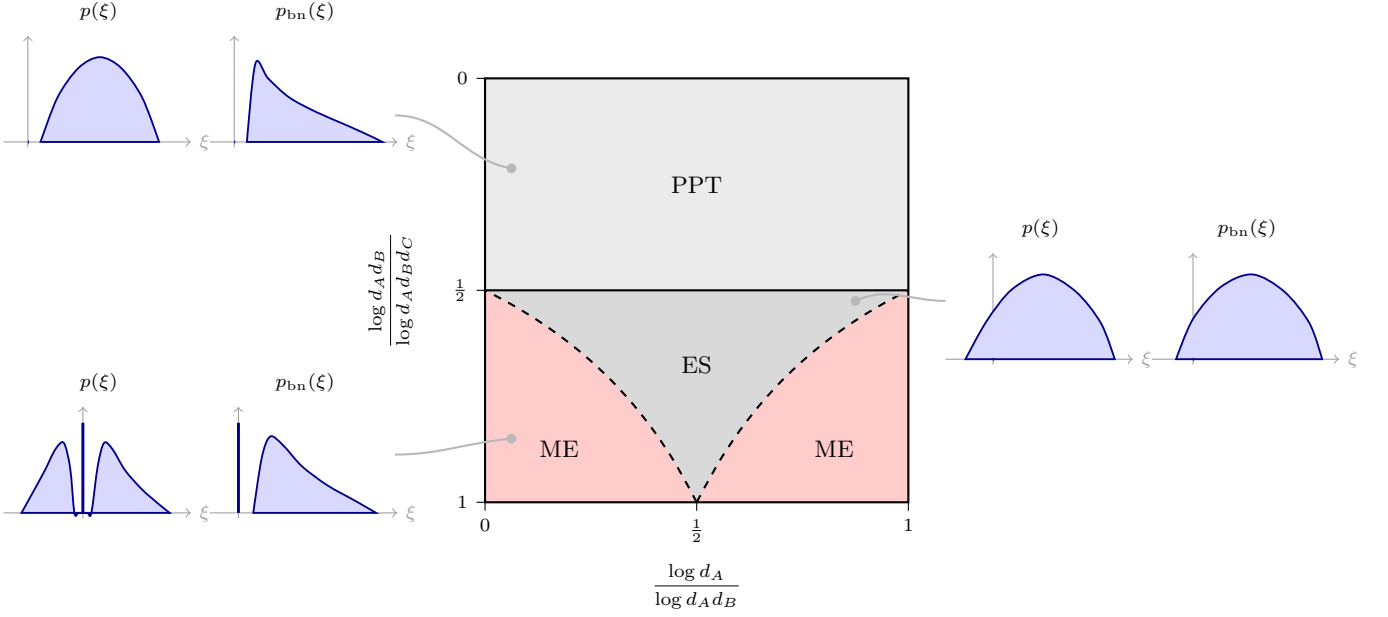
\begin{figure*}[t]
    \centering
    \begin{tikzpicture}[font=\small, scale=1.4, spinelabel/.style={font=\scriptsize},
        lead/.style={gray!55,line width=0.8pt}]
      \def\Rcurve{(2,0)(2.2,0.364)(2.4,0.667)(2.6,0.923)(2.8,1.143)(3.0,1.333)(3.2,1.5)(3.4,1.647)(3.6,1.778)(3.8,1.895)(4,2)}
      \def\Lcurve{(2,0)(1.8,0.364)(1.6,0.667)(1.4,0.923)(1.2,1.143)(1.0,1.333)(0.8,1.5)(0.6,1.647)(0.4,1.778)(0.2,1.895)(0,2)}
      \fill[black!8]  (0,2) rectangle (4,4);                  
      \fill[black!15] (0,0) rectangle (4,2);                  
      \fill[red!20] plot[smooth] coordinates \Rcurve -- (4,0) -- cycle;  
      \fill[red!20] plot[smooth] coordinates \Lcurve -- (0,0) -- cycle;  
      \draw[thick] (0,0) rectangle (4,4);
      \draw[thick] (0,2) -- (4,2);                            
      \draw[dashed,thick] plot[smooth] coordinates \Rcurve;   
      \draw[dashed,thick] plot[smooth] coordinates \Lcurve;
      \node[font=\small] at (2,3) {PPT};
      \node[font=\small] at (2,1.3) {ES};
      \node[font=\small] at (0.7,0.5) {ME};
      \node[font=\small] at (3.3,0.5) {ME};
      \draw (0,4)--(-0.08,4) node[left,spinelabel]{$0$};
      \draw (0,2)--(-0.08,2) node[left,spinelabel]{$\tfrac12$};
      \draw (0,0)--(-0.08,0) node[left,spinelabel]{$1$};
      \node[spinelabel,rotate=90,anchor=south] at (-0.6,2)
            {$\dfrac{\log d_Ad_B}{\log d_Ad_Bd_C}$};
      \draw (0,0)--(0,-0.08) node[below,spinelabel]{$0$};
      \draw (2,0)--(2,-0.08) node[below,spinelabel]{$\tfrac12$};
      \draw (4,0)--(4,-0.08) node[below,spinelabel]{$1$};
      \node[spinelabel,below] at (2,-0.5) {$\dfrac{\log d_A}{\log d_Ad_B}$};

      \specplot{-4.5}{3.4}{$p(\xi)$}{\xi}{0.18}
        {(0.30,0)(0.46,0.42)(0.66,0.70)(0.86,0.80)(1.06,0.70)(1.26,0.42)(1.42,0)}{0}
      \specplot{-2.55}{3.4}{$\pbn(\xi)$}{\xi}{0.18}
        {(0.30,0)(0.38,0.74)(0.50,0.60)(0.70,0.42)(0.96,0.28)(1.26,0.15)(1.50,0.04)(1.58,0)}{0}
      \specplot{-4.5}{-0.1}{$p(\xi)$}{\xi}{0.70}
        {(0.12,0)(0.22,0.18)(0.34,0.40)(0.44,0.60)(0.52,0.66)(0.58,0.42)(0.62,0)(0.66,0)(0.74,0)(0.78,0)(0.84,0.44)(0.90,0.66)(0.98,0.60)(1.10,0.40)(1.26,0.22)(1.42,0.08)(1.52,0)}{0.85}
      \specplot{-2.55}{-0.1}{$\pbn(\xi)$}{\xi}{0.22}
        {(0.36,0)(0.44,0.52)(0.52,0.72)(0.64,0.64)(0.82,0.44)(1.06,0.26)(1.33,0.11)(1.52,0)}{0.85}
      \specplot{4.4}{1.35}{$p(\xi)$}{\xi}{0.40}
        {(0.14,0)(0.34,0.36)(0.58,0.66)(0.88,0.80)(1.18,0.64)(1.42,0.34)(1.55,0)}{0}
      \specplot{6.35}{1.35}{$\pbn(\xi)$}{\xi}{0.34}
        {(0.18,0)(0.34,0.38)(0.60,0.68)(0.90,0.80)(1.20,0.64)(1.44,0.34)(1.56,0)}{0}

      \draw[lead] (-0.85,3.65) to[out=0,in=170] (.25,3.15);
      \fill[gray!55] (.25,3.15) circle (1.3pt);
      \draw[lead] (-0.85,0.45) to[out=0,in=185] (.25,0.6);
      \fill[gray!55] (.25,0.6) circle (1.3pt);
      \draw[lead] (4.35,1.9) to[out=180,in=25] (3.5,1.9);
      \fill[gray!55] (3.5,1.9) circle (1.3pt);
    \end{tikzpicture}
    \caption{Phase diagram of a random induced mixed state $\rho_{AB}$, parametrized by the relative subsystem sizes. The vertical axis $\log d_Ad_B/\log d_Ad_Bd_C$, which increases downward, measures how much of the global Hilbert space is retained in $AB$, and the horizontal axis $\log d_A/\log d_Ad_B$ measures the $A$/$B$ imbalance. The positive-partial-transpose (PPT), entanglement-saturation (ES), and maximally entangled (ME) phases are separated by the solid PPT/NPT line and the dashed ME boundaries. For each phase the left cartoon shows the spectral density $p(\xi)$ of the partial transpose $\rho_{AB}^{T_B}$ and the right cartoon the binegativity spectrum $\pbn(\xi)$ of $|\rho_{AB}^{T_B}|^{T_B}$. The binegativity is nonnegative except in the saturation phase, where a small tail crosses $\xi=0$.}
    \label{fig:phasediagram}
\end{figure*}

In this Letter, we demonstrate that the equivalence of negativity and PPT entanglement cost is generic. Namely, for random induced mixed states, $|\rho_{AB}^{T_B}|^{T_B}$ is either positive or close enough to positive such that the upper and lower bounds of \eqref{eq:bounds} collapse for large-dimensional Hilbert spaces. Equality is exact whenever the binegativity is positive semidefinite (the ME and PPT phases of figure \ref{fig:phasediagram}). In one phase (the ES phase of figure \ref{fig:phasediagram}), a small negative tail leaves a dimension-independent $O(1)$ gap $\log Z-\mathcal{E}$, asymptoting to $\approx0.22$ deep in the phase, so equality holds only to leading order.
This establishes the negativity as a precise operationally meaningful quantity for all but a measure zero portion of the Hilbert space, with respect to the Haar measure. In order to show this, we compute the full spectrum of the operator $|\rho_{AB}^{T_B}|^{T_B}$ and subsequently extract $\xi_{\min}$. This involves evaluating the moments  $\mathcal{Z}_{n,m}=\Tr\big(\big((\rho_{AB}^{T_B})^{n}\big)^{T_B}\big)^{m}$, which we prove are controlled by a two-parameter combinatorial sequence, the generalized Motzkin numbers \cite{sun2014congruences}. The generating function of these numbers encodes the spectrum of $|\rho_{AB}^{T_B}|^{T_B}$, and we verify the resulting spectrum using Monte-Carlo simulations.

\textit{Main result.}---
Our central claim rests, in the entanglement saturation phase, on hypotheses about the binegativity spectrum that we now isolate.
\begin{description}
  \item[(A) Replica continuation] The continuation $n\to1$ through even integers of the moments $\mathcal{Z}_{n,m}$ reproduces the binegativity moments $\langle\Tr(|\rho_{AB}^{T_B}|^{T_B})^m\rangle$.
  \item[(B) Edge convergence] The minimum eigenvalue converges to the lower edge of the binegativity spectrum $\pbn$, the distribution with moments $\langle\Tr(|\rho_{AB}^{T_B}|^{T_B})^m\rangle$.
  \item[(B$'$) One-sided edge bound] There is a constant $K$ with $\xi_{\min}(|\rho_{AB}^{T_B}|^{T_B})\ge -K/D$ with probability $\to1$. This is implied by (B) via the concentration of $\xi_{\min}$ (SM), and is all that the operational equality requires.
\end{description}

\begin{theorem}\label{thm:main}
Let $\rho_{AB}$ be a random induced mixed state drawn from a Haar random pure state on $\mathcal{H}_A \otimes \mathcal{H}_B \otimes \mathcal{H}_C$ with dimension $d_Ad_Bd_C$. As $d_A,d_B,d_C\to\infty$, with probability tending to one,
\begin{enumerate}
    \item[(i)] (PPT phase, $d_C>4d_Ad_B$) $\rho_{AB}^{T_B}\succeq0$, the binegativity equals $\rho_{AB}$, and $E_{ppt}=\mathcal{E}=0$.
        \item[(ii)] (ME phase, $d_{A/B}>d_{B/A}d_C$) $E_{ppt}=\mathcal{E}=\log d_{B/A}\,(1+o(1))$.
    \item[(iii)](ES phase, otherwise) Under \emph{(B$'$)}, the gap $\log Z-\mathcal{E}=O(1)$ throughout the phase. Deep in the phase ($d_Ad_B\gg d_C$), the negativity diverges, $\mathcal{E}=\tfrac12\log(d_Ad_B/d_C)+\log(8/3\pi)+o(1)\to\infty$, the gap is negligible, and $E_{ppt}/\mathcal{E}\to1$. Under \emph{(A)} and \emph{(B)} the gap is fixed at $\log Z-\mathcal{E}=\log\!\big[\,2\sqrt{1-\tfrac{64}{9\pi^2}}\,/\,(8/3\pi)\,\big]\approx0.22$.
\end{enumerate}

\end{theorem}

\noindent Parts \emph{(i)}, \emph{(ii)}, and the negativity of \emph{(iii)} are unconditional, resting on the (non-centered) semicircle and swap laws for $\rho_{AB}^{T_B}$~\cite{Aubrun2012,BanicaNechita2015,ShapourianEtAl2021} and, for \emph{(ii)}, the dimensional cost bound $E_{ppt}\le\log\min(d_A,d_B)$. The remainder constructs the binegativity spectrum underlying Theorem~\ref{thm:main}, proves the combinatorial identity it requires, and isolates the single estimate (B$'$) that would render the saturation phase fully unconditional. Technical proofs are given in the Supplemental Material (SM). There we also show that the negativity, the binegativity edge, and the gap self-average, concentrating about their means with vanishing fluctuations, so the equality holds for the typical state, not merely on average.

\textit{Negativity Spectrum.}---To induce a typical mixed state on $\mathcal{H}_A\otimes \mathcal{H}_B$, we can first take a Haar random pure state on tripartite Hilbert space $\mathcal{H}_A\otimes \mathcal{H}_B \otimes \mathcal{H}_C$
\begin{align}
    \ket{\Psi} = \sum_{i = 1}^{d_A d_B}\sum_{\alpha = 1}^{d_C} X_{i\alpha} \ket{\Psi^{(i)}_{AB}}\otimes \ket{\Psi^{(\alpha)}_C},
\end{align}
where $X_{i\alpha}$ are complex Gaussian random variables with joint probability distribution
\begin{align}
    P(\{ X_{i\alpha}\})\propto e^{-d_Ad_Bd_C\Tr(XX^{\dagger})}.
\end{align}
Tracing over $C$, we arrive at a random mixed state represented by density matrix $\rho_{AB} ={XX^{\dagger}}$.\footnote{To be a valid state, we normalize by $\Tr(XX^{\dagger})$ which induces the Hilbert-Schmidt measure on density matrices \cite{ZyczkowskiSommers2001}. The normalization is a global positive rescaling that preserves the sign of every binegativity eigenvalue, hence the entire phase structure and the binegativity-positivity question are normalization-independent, and shifts $\mathcal{E}$ by only $O(1/D)$, with $D=\sqrt{d_Ad_Bd_C}$. We make this rigorous in the SM.}

The moments of $\rho_{AB}^{T_B}$ are computed via sums over Wick contractions which can be organized as a sum over the permutation group $S_n$
\begin{align}
    \langle \Tr(\rho_{AB}^{T_B})^n\rangle = \frac{1}{(d_A d_Bd_C)^n}\sum_{\tau \in S_n} d_A^{\cyc(\eta^{-1} \tau)}d_B^{\cyc(\eta \tau)}d_C^{\cyc( \tau)},
\end{align}
where $\eta$ is the cyclic permutation and $\cyc(\cdot)$ counts the number of cycles in the permutation. Depending on the relative sizes of the subsystems, different permutations dominate the sum. As shown in \cite{Aubrun2012,BanicaNechita2015,banica2013asymptotic,ShapourianEtAl2021}, there are three distinct phases.
\begin{enumerate}
    \item Positive Partial Transpose phase ($d_C > 4d_Ad_B$). The bath ($\mathcal{H}_C$) is sufficiently large, such that there is no quantum entanglement present. The spectrum of $\rho_{AB}^{T_B}$ is a semicircle
    \begin{align}
        p(\xi) = \frac{d_Ad_Bd_C}{2\pi}\sqrt{\frac{4}{d_Ad_Bd_C} - \left( \xi- \frac{1}{d_Ad_B}\right)^2}.
        \label{eq:semicircle}
    \end{align}
    \item Maximally Entangled phase ($d_A > d_Bd_C$ or $d_B > d_Ad_C$). One party dominates the others. Taking $d_A > d_B d_C$ (the case $A \leftrightarrow B$ is similar), the negativity saturates to its maximal value $\log d_B$. The nonzero eigenvalues form two Marchenko--Pastur (MP) clusters. The positive cluster is
    \begin{align}
        p_+(\xi) &= \frac{d_Bd_C\sqrt{(\xi_+-\xi)(\xi-\xi_-)}}{2\pi q\,\,\xi},
    \end{align}
    with edges $\xi_\pm = (d_Bd_C)^{-1}(1\pm\sqrt{q})^2$ and aspect ratio $q = d_Bd_C/(2d_A)$. The negative cluster is its reflection $\xi \to -\xi$. The two clusters carry $N_\pm = \tfrac12 d_Bd_C(d_B\pm1)$ eigenvalues, and the imbalance $N_+-N_- = d_Bd_C$ enforces $\Tr\rho_{AB}^{T_B}=1$.
    \item Entanglement Saturation phase (otherwise). The spectrum is the same semicircle of \eqref{eq:semicircle}, though there now exist negative eigenvalues, so the negativity is nontrivial. Deep in the phase ($d_Ad_B\gg d_C$),
    \begin{align}
        \langle \mathcal{E}(\rho_{AB})\rangle = \frac{1}{2}\log \left( \frac{d_A d_B}{d_C}\right) + \log \left(\frac{8}{3\pi}\right).
    \end{align}
\end{enumerate}
The phase diagram is summarized in figure \ref{fig:phasediagram}.

\textit{Binegativity Spectrum.}---
We now turn to the evaluation of the spectrum of $|\rho_{AB}^{T_B}|^{T_B}$, the binegativity operator. This is made possible using a double replica trick
\begin{align}
     \mathcal{Z}_{n,m} = \Tr (((\rho_{AB}^{T_B})^{n} )^{T_B} )^m
\end{align}
taking even values of $n$ to $1$ to implement the absolute value. We can label the replicas with two indices $(j,k)$ with $0\leq j<  m$ and $0\leq k < n$.
This may also be evaluated for random mixed states using a sum of permutations
\begin{align}
     \langle \mathcal{Z}_{n,m}\rangle = \frac{1}{(d_A d_B d_C)^{nm}} \sum_{\tau \in S_{nm}} d_A^{\cyc ( \eta^{-1}\tau )} d_B^{\cyc (\Sigma_B^{-1}  \tau)} d_{C}^{\cyc (\tau)},
\end{align}
where $\Sigma_B$ acts via
\begin{equation}\label{eq:SigmaB}
  \Sigma_B(j,k) = (j,\,k-1) \quad (k \ge 1), \quad
  \Sigma_B(j,0) = (j+1,\,n-1).
\end{equation}

The mean ($m = 1$) for the binegativity is equal to that of the negativity because the partial transpose is trace preserving, $\Tr(|\rho_{AB}^{T_B}|^{T_B})= \Tr(|\rho_{AB}^{T_B}|)$. For the second moment ($m=2$), a partial-transpose Hilbert-Schmidt identity gives $\Tr((|\rho_{AB}^{T_B}|^{T_B})^2) = \Tr((\rho_{AB}^{T_B})^2) = \Tr(\rho_{AB}^2)$, whose ensemble average is $\langle\Tr(\rho_{AB}^2)\rangle = \tfrac{1}{d_Ad_B} + \tfrac{1}{d_C}$ at leading order in every phase.

In the PPT phase, $|\rho_{AB}^{T_B}| = {\rho_{AB}^{T_B}}$ by definition. Thus, $|\rho_{AB}^{T_B}|^{T_B} = \rho_{AB}$, which is manifestly positive.
The binegativity spectrum is that of $\rho_{AB}$, which is a Wishart random matrix and thus asymptotically follows an MP distribution
\begin{align}
\begin{aligned}
    \pbn(\xi) = \frac{d_Ad_B\sqrt{(\xi_+-\xi)(\xi -\xi_-)}}{2\pi q\xi}, \\ \xi_{\pm} = (d_Ad_B)^{-1}(1\pm\sqrt{q})^2,
\end{aligned}
\end{align}
where $q = \tfrac{d_Ad_B}{d_C}$.
Thus $Z(\rho_{AB}) = 1$ and
\begin{align}
    \langle E_{ppt}(\rho_{AB} )\rangle = \langle \mathcal{E}(\rho_{AB})\rangle = 0.
\end{align}

For the maximally entangled phase, first consider $d_A \gg d_B  d_C$. In this case, $\tau = \eta$ dominates the sum. For even $n$, $\Sigma_B^{-1} \eta$ has two cycles, so
\begin{align}
    \langle \mathcal{Z}_{n,m}\rangle = d_B^{2-nm}d_C^{1-nm},
\end{align}
which leads to a spectrum of $d_B^2 d_C$ eigenvalues at $\tfrac{1}{d_Bd_C} $ with the rest at zero. When $d_A$ becomes of order $d_B  d_C$, the sharp peak of eigenvalues at $\tfrac{1}{d_Bd_C} $ broadens while the zero eigenvalues are robust. We conjecture that the broadening is a MP distribution. The aspect ratio is then fixed by the first two moments to $q = \tfrac{d_B d_C}{2d_A}$. An MP distribution then predicts skewness $\sqrt{q}$ and kurtosis $ q+2$, which we check numerically at $q= \tfrac{1}{4}$
\begin{center}\small
\begin{tabular}{lccccc}
\toprule
$d_B^2 d_C$ & 64 & 128 & 512 & 1024 & 4096 \\
\hline
skewness  & 0.585 & 0.576 & 0.515 & 0.511 & 0.492 \\
kurtosis  & 2.498 & 2.467 & 2.307 & 2.294 & 2.244 \\
\hline
\end{tabular}\\[2pt]
{\footnotesize MP predictions with skewness $\sqrt{q}=0.500$ and kurtosis $q+2=2.250$.}
\end{center}
Irrespective of the validity of this conjecture, the binegativity spectrum appears nonnegative throughout this phase numerically, and at leading order, unconditionally,
\begin{align}
     \langle E_{ppt}(\rho_{AB} )\rangle = \langle \mathcal{E}(\rho_{AB})\rangle = \log d_B\,(1+o(1)).
\end{align}

\begin{figure}
    \centering
    \includegraphics[width=1\linewidth]{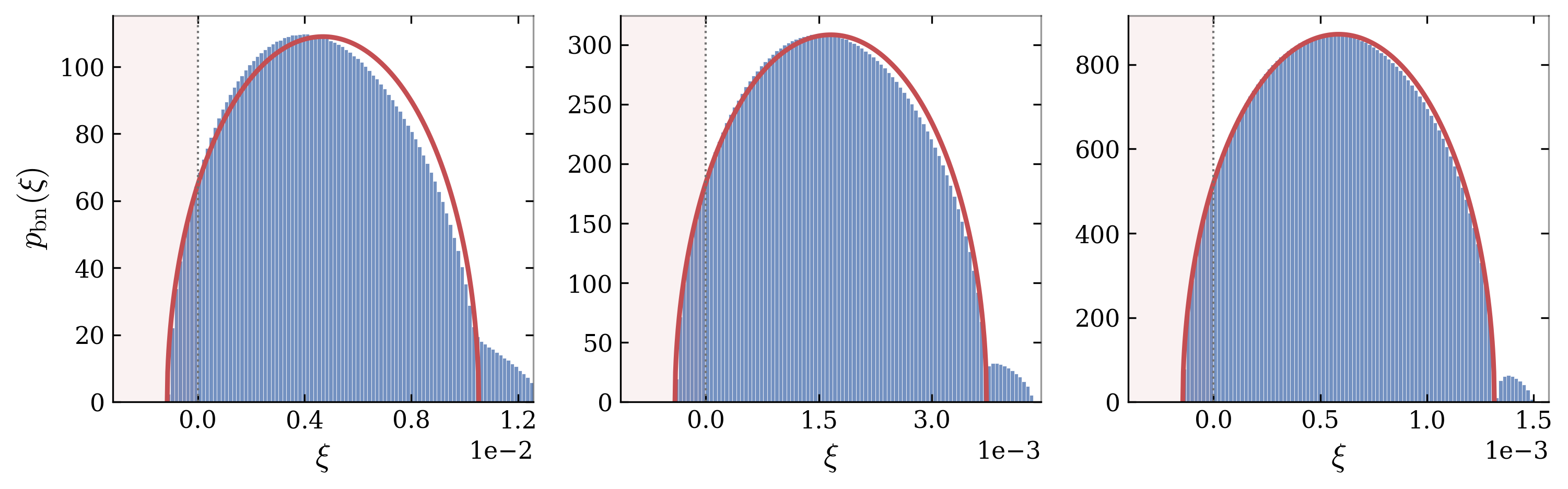}
    \caption{Entanglement-saturation-phase binegativity spectrum $\pbn(\xi)$ at $d_A=d_B=d_C=2^N$ for $N=5,6,7$ (left to right, with $5000$, $1500$, and $300$ Haar samples). The red curve is the semicircle law \eqref{eq:binegsemi}, with no fit parameters. The shaded tail to the left of $\xi=0$ is the negative support responsible for the $O(1)$ gap $\log Z - \mathcal{E}$.}
    \label{fig:semicircle}
\end{figure}
When the three subsystems are of comparable size, deep inside the entanglement saturation phase, we seek permutations that maximize the sum $P(\tau) \equiv {\cyc ( \eta^{-1}\tau )} +{\cyc (\Sigma_B^{-1}  \tau)} +{\cyc (\tau)}$. The minimal number of transpositions composing a permutation $\sigma$ is $|\sigma| = |\sigma^{-1}| = nm-\cyc(\sigma)$ and this is a metric on the permutation group, so obeys a triangle inequality $|\alpha \beta|\leq |\alpha|+ |\beta|$. Using this, $|\tau| + |\tau^{-1} \eta|\geq nm -1$, $|\tau| + |\Sigma_B^{-1}\tau|\geq nm-1$, and $|\eta^{-1}\tau| + |\Sigma_B^{-1} \tau| \geq nm-2$ (for even $n$, $\eta^{-1}\Sigma_B$ has two cycles), which together lead to the bound $P(\tau ) \leq \tfrac{3}{2}nm+2$. We find that not only is this bound saturated, but there is a large degeneracy of permutations, $D(n,m)$, that saturate it. For $m = 1$, these are the set of non-crossing pairings identified in \cite{Aubrun2012} and given by the Catalan number $C_{n/2} = \frac{1}{n/2+1}\binom{n}{n/2}$, and for $m = 2$, this is $C_n$. We verified this exhaustively up to $nm=12$, and prove in general (SM) that the saturating permutations are precisely the perfect matchings non-crossing in both the forward ($\eta$) and reversed-block ($\Sigma_B$) cyclic orders, yielding the closed form
\begin{align}
    D(n,m) = M_m(C_{n/2}, C_n -C_{n/2}^2),
    \label{eq:Dnm}
\end{align}
where $M_m(b,c)$ are the generalized Motzkin numbers \cite{sun2014congruences}
\begin{align}
    M_m (b, c) = \sum_{k = 0}^{\lfloor m/2 \rfloor} \begin{pmatrix}
m \\
2k
\end{pmatrix} C_k b^{m - 2k} c^k
\end{align}
The corresponding generating function is \cite{Wang2015Combinatorics}
\begin{align}
\sum_{m \ge 0} M_m (b, c) x^m = \frac{1 - bx - \sqrt{(1 - bx)^2 - 4 c x^2}}{2 c x^2}.
\end{align}
Combinatorially, this counts the non-crossing pairings such that each block of $m$ either closes on itself (giving $C_{n/2}$) or is paired (in a non-crossing fashion) with one other block (giving $C_n - C_{n/2}^2$). For $m=1$ these reproduce the moments of a shifted semicircle law. Continuing $n\to1$, the mean is $C_{1/2} = 8/(3\pi)$ and the variance is $\sigma^2 = 1-C_{1/2}^2 = 1-64/(9\pi^2)$, so that, with $D = \sqrt{d_Ad_Bd_C}$,
\begin{align}
    \pbn(\xi) = \frac{D}{2\pi\sigma^2}
    \sqrt{4\sigma^2 -\left(D\xi-\tfrac{8}{3\pi}\right)^2}.
\label{eq:binegsemi}
\end{align}
The semicircle shape now follows from this count under the replica continuation (A), and its mean and variance are independently rigorous, being the negativity and the Hilbert--Schmidt second moment. We numerically verify that the spectrum converges to this semicircle in the large dimension limit (see figure \ref{fig:semicircle}).

The left edge of the semicircle is at $D\xi_{\min}
=
\frac{8}{3\pi}
-
2\sqrt{1-\frac{64}{9\pi^2}}
\approx -0.209$. Thus, the binegativity develops a small negative tail. On universality grounds, we expect the tail to be controlled by Tracy-Widom fluctuations in the GUE symmetry class \cite{TracyWidom1994}. These fluctuations are subleading in the large dimension limit, which we have numerically verified, so at leading order, we find
\begin{align}
    \langle Z(\rho_{AB})\rangle =2\sqrt{1-\frac{64}{9\pi^2}}\;\sqrt{\frac{d_Ad_B}{d_C}}.
\end{align}
The gap between $\log Z(\rho_{AB})$ and $\mathcal{E}(\rho_{AB})$ is therefore $O(1)$ and so the relative correction vanishes, $E_{ppt}/\mathcal{E} \rightarrow 1$, and equality holds at leading order. Across the entire phase diagram, the negativity thus equals the PPT entanglement cost, exactly where the binegativity is positive, and up to a vanishing relative correction in the entanglement saturation phase.

\textit{Phase Boundaries.}---There are two qualitatively different entanglement phase transitions. As one crosses from the ES phase to the ME phase, $\rho_{AB}^{T_B}$ becomes rank deficient, so the phase transition is sharp and the left edge of the binegativity spectrum jumps discontinuously from a negative value to zero. The transition from the PPT phase to the ES phase, by contrast, is smooth. When $d_C = 4 d_A d_B$, the semicircle \eqref{eq:semicircle} gains support on the negative real axis, so the negativity becomes nonzero. We find numerically that the threshold ratio $d_C/(d_Ad_B)$ at which the binegativity left edge $D\xi_{\min}$ crosses zero slowly approaches a number less than $4$ that we optimistically conjecture is given by $C_{1/2}=8/(3\pi)$ as $d_A,d_B\to\infty$, so the binegativity-violation threshold is $d_C \approx \tfrac{8}{3\pi}\,d_Ad_B$. The table below shows the measured ratio $d_C/(d_Ad_B)$ at fixed $d_B=4$
\begin{center}\small
\begin{tabular}{lccccc}
\toprule
$d_A$ & $128$ & $256$ & $512$ & $1024$ & $2048$ \\
\hline
$d_C/(d_Ad_B)$ & 0.791 & 0.808 & 0.819 & 0.826 & 0.830 \\
\hline
\end{tabular}\\[2pt]
{\footnotesize Conjectured asymptotic value $C_{1/2}=8/(3\pi)\approx0.849$.}
\end{center}
There is a finite window in which $\rho_{AB}^{T_B}$ has negative eigenvalues yet $|\rho_{AB}^{T_B}|^{T_B}$ remains positive, so $E_{ppt} = \mathcal{E}$ holds exactly even though the entanglement is nonzero and $O(1)$.

\textit{Discussion.}---In this Letter, we have demonstrated that for large random induced mixed states, the logarithmic negativity equals the exact PPT entanglement cost. In the PPT phase both quantities vanish. In the ME phase $E_{ppt} = \mathcal{E} = \log d_B\,(1+o(1))$. Deep in the ES phase ($d_Ad_B\gg d_C$), the $O(1)$ gap $\log Z - \mathcal{E} \approx 0.22$ (under (A),(B)) is negligible compared to the $\tfrac{1}{2}\log(d_Ad_B/d_C)$ negativity, so $E_{ppt}/\mathcal{E} \to 1$. The PPT phase is rigorous and exact; the ME phase is rigorous to leading order. The saturation phase rests on the spectral hypotheses (A) and (B), whose qualitative content is captured by the single edge bound (B$'$). These are natural hypotheses and have strong numerical support. Together, these establish negativity as an operationally sharp entanglement measure for generic quantum states.

The logarithmic negativity is not only efficiently computable but experimentally accessible. The entanglement transition of pseudo-random mixed states has been observed on a superconducting quantum processor~\cite{LiuEtAl2023}, and the partial-transpose moments needed to locate the phase diagram are accessible through randomized measurements on quantum simulators~\cite{CarrascoEtAl2024}. The binegativity moments are similarly accessible, and numerical analytic-continuation methods developed for the randomized-measurement toolbox~\cite{VijayEtAl2025} can be applied to the replica index $\mathcal{Z}_{n,m}$. The operational equivalence we establish gives these experiments a precise resource-theoretic meaning.

Finally, we note that the technical tools underlying our calculation, namely a sum over $S_{nm}$ with three cycle-count weights, are not special to Haar random induced mixed states. The same structure governs random tensor network states~\cite{HaydenNezami2016,DongQiWalter2021,KudlerFlamNarovlanskyRyu2022} and fixed-area states in AdS/CFT~\cite{AkersRath2019,DongHarlowMarolf2019}, where holographic negativity is computed by analogous replica methods and our analysis extends straightforwardly. It also controls entanglement in random quantum circuits, including measurement-induced phase transitions~\cite{SangEtAl2021,WeinsteinBaoAltman2022}. For models of evaporating black holes, similar permutation considerations yield the same three phases in the negativity spectrum~\cite{DongMcBrideWeng2022}, where we likewise expect $E_{ppt} = \mathcal{E}$. Whether this equality persists for more general holographic states~\cite{DongKudlerFlamRath2025}, or for high-energy many-body pure states where additional phases appear~\cite{VardhanEtAl2022,VardhanEtAl2023}, is an important open question.

\begin{acknowledgements}
 \textit{Acknowledgments.}---We thank Pratik Rath for initial collaboration and comments. We thank Francesco Anna Mele for correspondence that motivated this work, for helpful comments, and providing a proof of \eqref{eq:Dnm}, which was only a conjecture in a previous draft. JKF is supported by the Marvin L. Goldberger Member Fund at the Institute for Advanced Study and the National Science Foundation under Grant No. PHY-2514611.
\end{acknowledgements}
\bibliographystyle{apsrev4-2}
\bibliography{main}

\clearpage
\onecolumngrid
\setcounter{secnumdepth}{3}
\setcounter{section}{0}
\setcounter{equation}{0}
\setcounter{figure}{0}
\setcounter{table}{0}
\setcounter{theorem}{0}
\setcounter{proposition}{0}
\setcounter{lemma}{0}
\setcounter{corollary}{0}
\setcounter{remark}{0}
\renewcommand{\thesection}{S\arabic{section}}
\renewcommand{\theequation}{S\arabic{equation}}
\renewcommand{\thefigure}{S\arabic{figure}}
\renewcommand{\thetable}{S\arabic{table}}
\renewcommand{\thetheorem}{S\arabic{theorem}}
\renewcommand{\theproposition}{S\arabic{proposition}}
\renewcommand{\thelemma}{S\arabic{lemma}}
\renewcommand{\thecorollary}{S\arabic{corollary}}
\renewcommand{\theremark}{S\arabic{remark}}

\begin{center}
{\large\textbf{Supplemental Material}}\\[3pt]
\end{center}
\vspace{0.5em}

This Supplemental Material provides technical details supporting the main text. We mark each result as \emph{rigorous}, \emph{conditional} (on hypotheses (A), (B), (B$'$)), or \emph{numerical}. Sections~\ref{sec:ape}--\ref{sec:pmax} are rigorous. Section~\ref{sec:motzkin} proves the dominant-permutation count. Section~\ref{sec:edge} treats the spectral edge and a numerical self-averaging remark. Throughout, $D=\sqrt{\dA\dB\dC}$ and $\Chalf=8/(3\pi)$ is the half-integer Catalan number.

\section{The operational reduction}\label{sec:ape}

Let $\rho_{AB}$ be a state on $\mathcal{H}_A\otimes\mathcal{H}_B$ with dimension $\dA\dB$, and let $E_{ppt}$ denote the exact entanglement cost under PPT-preserving operations. Audenaert, Plenio and Eisert \cite{APE2003} proved
\begin{equation}\label{eq:apebound}
  \mathcal{E}(\rho_{AB})\ \le\ E_{ppt}(\rho_{AB})\ \le\ \log Z(\rho_{AB}),
  \qquad
  Z(\rho_{AB}) = \|\rhoTB\|_1 + \dA\dB\,\max\!\big(0,\,-\xi_{\min}(\bineg)\big),
\end{equation}
where $\mathcal{E}=\log\|\rhoTB\|_1$ and $\bineg$ is the binegativity operator \cite{Ishizaka2004}. The lower bound is the negativity, and the upper bound exceeds it only through the negative part of the binegativity spectrum.

\begin{corollary}[Operational reduction]\label{cor:reduction}
If $\bineg\succeq0$ then $Z=\|\rhoTB\|_1$ and $E_{ppt}(\rho_{AB})=\mathcal{E}(\rho_{AB})$ exactly. More generally, if $-\xi_{\min}(\bineg)\le K/D$ for a constant $K$, then throughout the saturation phase, where $\|\rhoTB\|_1\ge\Chalf\sqrt{\dA\dB/\dC}$ (Sec.~\ref{sec:negativity}),
\begin{equation}
  0\ \le\ \log Z-\mathcal{E}\ \le\ \log\!\Big(1+\tfrac{K}{\Chalf}\Big),
\end{equation}
so the gap is $O(1)$.
\end{corollary}
\begin{proof}
The first claim is immediate. For the second, $\dA\dB(-\xi_{\min})\le K\sqrt{\dA\dB/\dC}\le(K/\Chalf)\|\rhoTB\|_1$, using the stated lower bound, so $Z\le(1+K/\Chalf)\|\rhoTB\|_1$. Taking logs, the gap is $O(1)$; in the deep-ES regime $\dA\dB/\dC\to\infty$ where $\mathcal{E}\to\infty$, the relative gap vanishes and $E_{ppt}/\mathcal{E}\to1$.
\end{proof}

Thus the operational question reduces entirely to the spectrum of $\bineg$, and in particular to $\xi_{\min}(\bineg)$.

\section{Replica permutation sums and normalization}\label{sec:perm}

A Haar-random pure state on $\mathcal{H}_A\otimes\mathcal{H}_B\otimes\mathcal{H}_C$ gives $\rho_{AB}=XX^\dagger/\Tr(XX^\dagger)$ with $X\in\mathbb{C}^{\dA\dB\times\dC}$ i.i.d.\ complex Gaussian. Wick contraction of the unnormalized $W=XX^\dagger$ gives, for the negativity moments, the single-replica sum
\begin{equation}\label{eq:single}
  \big\langle\Tr(\rhoTB)^n\big\rangle = \frac{1}{(\dA\dB\dC)^n}\sum_{\tau\in S_n}\dA^{\cyc(\eta^{-1}\tau)}\dB^{\cyc(\eta\tau)}\dC^{\cyc(\tau)},
\end{equation}
with $\eta$ the full $n$-cycle cyclic permutation, and for the binegativity the double-replica moment $\mathcal{Z}_{n,m}=\Tr(((\rhoTB)^{n})^{T_B})^{m}$,
\begin{equation}\label{eq:masterS}
  \big\langle\mathcal{Z}_{n,m}\big\rangle = \frac{1}{(\dA\dB\dC)^{nm}}\sum_{\tau\in S_{nm}}\dA^{\cyc(\eta^{-1}\tau)}\dB^{\cyc(\Sigma_B^{-1}\tau)}\dC^{\cyc(\tau)},
  \qquad
  \Sigma_B(j,k)=(j,k-1)\ (k\ge1),\ \ \Sigma_B(j,0)=(j+1,n-1),
\end{equation}
with $\eta$ the full $nm$-cycle. In index notation $(r,s)$, this should be understood as $(r\mod m, s \mod n)$. For $m=1,2$ one has $\Sigma_B=\eta^{-1}$ (the $m=2$ coincidence is the Hilbert--Schmidt identity, Lemma~\ref{lem:HS}). We validated \eqref{eq:masterS} against direct Monte-Carlo Gaussian averages at $(n,m)=(2,2),(3,2),(2,3),(4,1)$.

\paragraph{Normalization (rigorous).} The division by $\Tr W$ is harmless, as follows.
\begin{lemma}\label{lem:norm}
$\langle\Tr W\rangle=1$ and $\Pr(|\Tr W-1|>t)\le 2e^{-cD^2\min(t,t^2)}$ for an absolute $c>0$, so $\Tr W=1+O(D^{-1}\sqrt{\log D})$ with probability $\to1$. Writing $\rho_{AB}=W/\Tr W$, \emph{(i)} every eigenvalue satisfies $\xi_i(\rho_{AB})=\xi_i(W)/\Tr W$ with $\Tr W>0$, so the sign of each binegativity eigenvalue, and hence the phase boundaries and the binegativity-positivity question, is exactly independent of the normalization. \emph{(ii)} $Z(\rho_{AB})=Z(W)/\Tr W$ and $\|\rho_{AB}^{T_B}\|_1=\|W^{T_B}\|_1/\Tr W$, so the gap $\log Z-\mathcal{E}$ is exactly normalization-free, while $\mathcal{E}(\rho_{AB})=\mathcal{E}(W)-\log\Tr W=\mathcal{E}(W)+O(D^{-1}\sqrt{\log D})$.
\end{lemma}
\begin{proof}
$\Tr W=\sum_{a,b=1}^{D^2}|X_{ab}|^2= D^{-2}\sum_{a,b} E_{ab}$ where $D^2=\dA\dB\dC$ is the number of entries and $E_{ab}\equiv D^2|X_{ab}|^2$ are i.i.d.\ unit-mean exponentials, so $\langle\Tr W\rangle=1$. Bernstein's inequality for the sub-exponential $E_a-1$ gives the tail (\cite{vershynin2018high} Theorem 2.8.1), with $\Tr W=1+O(D^{-1}\sqrt{\log D})$ with high probability. For the rest, $|cM|^{T_B}=c\,|M|^{T_B}$ and $\|cM\|_1=c\|M\|_1$ for $c=1/\Tr W>0$, and $\mathcal{E}=\log\|\cdot\|_1$, $\log Z$ each pick up the same $-\log\Tr W$.
\end{proof}
\noindent Thus the normalization is irrelevant to every sign-dependent statement (phases, positivity, $\xi_{\min}\gtrless0$) and to the gap, and contributes only a vanishing $O(1/D)$ to the negativity. We may therefore work with $W$ at leading order.

\section{The negativity}\label{sec:negativity}

The negativity is the rigorous lower bound of \eqref{eq:apebound}, resting on the (non-centered) semicircle law for $\rhoTB$ \cite{Aubrun2012,BanicaNechita2015,FukudaSniady2013}. With $c=\dC/(\dA\dB)$, the rescaled eigenvalues of $\rhoTB$ converge to a semicircle of mean $1/(\dA\dB)$ and radius $2/D$. It reaches the negative axis once $2/\sqrt c>1$, fixing the PPT/NPT threshold at $c=4$ exactly. Integrating $|\lambda|$ over the shifted semicircle gives, for $0<c\le4$,
\begin{equation}\label{eq:neg-closed}
  \|\rhoTB\|_1=\frac{2}{\pi}\arcsin\frac{\sqrt c}{2}+\frac{\sqrt{c(4-c)}}{2\pi}+\frac{(4-c)^{3/2}}{3\pi\sqrt c},
\end{equation}
equal to $1$ for $c\ge4$. Deep in saturation ($c\ll4$) the mean shift is negligible against the radius $R=2/D$, and with $\langle|\lambda|\rangle=\tfrac{4}{\pi R^2}\int_0^R\lambda\sqrt{R^2-\lambda^2}\,d\lambda=\tfrac{4R}{3\pi}$ and $\dA\dB$ eigenvalues,
\begin{equation}\label{eq:neg-deep}
  \|\rhoTB\|_1=\dA\dB\cdot\frac{4R}{3\pi}=\frac{8}{3\pi}\sqrt{\frac{\dA\dB}{\dC}},
  \qquad
  \mathcal{E}=\frac12\log\frac{\dA\dB}{\dC}+\log\frac{8}{3\pi},
\end{equation}
the constant being $\log\Chalf$ (Eq.~(8) of the main text), the same $\Chalf=8/(3\pi)$ that centers the binegativity semicircle.

\section{PPT phase: exact Marchenko--Pastur}\label{sec:ppt}

\begin{proposition}[Part (i)]\label{prop:ppt}
For $\dC>4\dA\dB$, with probability $\to1$ one has $\rhoTB\succeq0$, hence $\bineg=\rho_{AB}\succeq0$, and $E_{ppt}=\mathcal{E}=0$.
\end{proposition}
\begin{proof}
By \cite{Aubrun2012} the least eigenvalue of $\rhoTB$ is nonnegative with high probability once $c>4$. Then $|\rhoTB|=\rhoTB$, so $\bineg=(\rhoTB)^{T_B}=\rho_{AB}\succeq0$, a Marchenko--Pastur law with $q=\dA\dB/\dC$. Positivity gives $Z=\|\rhoTB\|_1=1$, $\mathcal{E}=0$, and Corollary~\ref{cor:reduction} gives $E_{ppt}=\mathcal{E}=0$.
\end{proof}

\section{Maximally entangled phase: the swap limit and a dimensional bound}\label{sec:maxent}

\begin{proposition}[Part (ii), leading order]\label{prop:maxent}
For $\dA\gg\dB\dC$, $E_{ppt}=\mathcal{E}=\log\dB\,(1+o(1))$.
\end{proposition}
\begin{proof}
Index $X_{(ab),c}$ and let $x^{(b,c)}\in\mathbb{C}^{\dA}$ have entries $x^{(b,c)}_a=X_{(ab),c}$, giving $\dB\dC$ i.i.d.\ complex Gaussian vectors. The partial transpose decomposes into $\dB\times\dB$ blocks $[\rhoTB]^{(b,b')}=Z^{-1}\sum_c x^{(b',c)}(x^{(b,c)})^\dagger$. As $\dA\to\infty$ the vectors become orthonormal and $Z=\Tr(XX^\dagger)\to\dA\dB\dC$, so on its support $\rhoTB\to\tfrac{1}{\dB\dC}\mathbb{S}$, the swap on $\mathbb{C}^{\dB}\otimes\mathbb{C}^{\dB}$. Since $\mathbb{S}$ has eigenvalues $\pm1$, $|\rhoTB|\to\tfrac{1}{\dB\dC}\Pi$ ($\Pi$ the support projector) and
\begin{equation}\label{eq:swap-spectrum}
  \bineg\to\tfrac{1}{\dB\dC}\Pi^{T_B}=\tfrac{1}{\dB\dC}\,\mathbb{I}_{\rm supp}\succeq0.
\end{equation}
The leading-order $\rhoTB$ spectrum is $\pm\tfrac1{\dB\dC}$ with multiplicities $\dC\tfrac{\dB(\dB\pm1)}{2}$ and a zero of multiplicity $\dB(\dA-\dB\dC)$, confirmed to the integer (e.g.\ $(\dA,\dB,\dC)=(256,2,2)$ gives $N_+=6,N_-=2,N_0=504$). Taking $|\cdot|$ and $T_B$ gives \eqref{eq:swap-spectrum}, so $\bineg\succeq0$ at leading order and $\mathcal{E}=\log\|\rhoTB\|_1\to\log\dB$.

\emph{Rigorous equality via a dimensional bound.} The upper bound on need not use the binegativity at all. $E_{ppt}\leq \log\min(\dA,\dB) = d_B$ for all states (here $\dA>\dB\dC\ge\dB$). With the rigorous lower bound $E_{ppt}\ge\mathcal{E}=\log\dB(1-o(1))$,
\begin{equation}
  \log\dB(1-o(1))\ \le\ E_{ppt}\ \le\ \log\dB,
\end{equation}
so $E_{ppt}=\log\dB(1-o(1))$ \emph{unconditionally}, no edge control of the binegativity is required in this phase.
\end{proof}
\noindent The broadening of the $\tfrac1{\dB\dC}$ peak at finite $\dA/\dB\dC$ is numerically Marchenko--Pastur with $q_{\rm peak}=\dB\dC/(2\dA)$ (fixed by the first two moments). This fine structure does not affect the operational conclusion. Numerically the binegativity is nonnegative throughout the phase, so the equality appears to hold exactly.

\section{Second moment and the maximal cycle count}\label{sec:pmax}

\begin{lemma}[Partial-transpose Hilbert--Schmidt identity]\label{lem:HS}
For any $A,B$ on $\mathcal{H}_A\otimes\mathcal{H}_B$, $\Tr(A^{T_B}B^{T_B})=\Tr(AB)$.
\end{lemma}
\begin{proof}
$(A^{T_B})_{(i,j),(k,l)}=A_{(i,l),(k,j)}$, so $\Tr(A^{T_B}B^{T_B})=\sum_{i,j,k,l}A_{(i,l),(k,j)}B_{(k,j),(i,l)}=\Tr(AB)$.
\end{proof}

\begin{theorem}[Second moment]\label{thm:alpha2}
For the unnormalized $W=XX^\dagger$, $\big\langle\Tr(\bineg)^2\big\rangle=\dfrac{1}{\dA\dB}+\dfrac{1}{\dC}$ exactly. For the normalized state $\rho_{AB}=W/\Tr W$, the exact value is the Lubkin formula $(\dA\dB+\dC)/(\dA\dB\dC+1)$~\cite{ZyczkowskiSommers2001}, which agrees to leading order.
\end{theorem}
\begin{proof}
By Lemma~\ref{lem:HS} with $A=B=|\rhoTB|$, $\Tr((\bineg)^2)=\Tr(|\rhoTB|^2)=\Tr((\rhoTB)^2)$. For $W$, the $n=2,m=1$ case of \eqref{eq:single} gives $1/(\dA\dB)+1/\dC$ exactly. For $\rho_{AB}=W/\Tr W$ the same HS identity reduces the second moment to the purity $\langle\Tr(\rho_{AB}^2)\rangle$, whose exact value is the Lubkin formula \cite{ZyczkowskiSommers2001}.
\end{proof}

\begin{proposition}[Maximal cycle count]\label{prop:pmax}
For every even $n$, $m$, and $\tau\in S_{nm}$, with $P(\tau)=\cyc(\tau)+\cyc(\eta^{-1}\tau)+\cyc(\Sigma_B^{-1}\tau)$, one has $P(\tau)\le\tfrac32nm+2$, and the bound is attained.
\end{proposition}
\begin{proof}
With $|\sigma|=nm-\cyc(\sigma)$ (a metric), $|\eta|=|\Sigma_B|=nm-1$, and $|\tau|+|\eta^{-1}\tau|\ge nm-1$, $|\tau|+|\Sigma_B^{-1}\tau|\ge nm-1$, $|\eta^{-1}\tau|+|\Sigma_B^{-1}\tau|\ge|\eta^{-1}\Sigma_B|$. For even $n$, $\eta^{-1}\Sigma_B$ has two cycles (even/odd positions), so $|\eta^{-1}\Sigma_B|=nm-2$. Summing, $2(|\tau|+|\eta^{-1}\tau|+|\Sigma_B^{-1}\tau|)\ge3nm-4$, i.e.\ $P(\tau)\le\tfrac32nm+2$. The attainment of the bound is the subject of the following section.
\end{proof}

\section{The generalized Motzkin count}\label{sec:motzkin}

We prove the dominant-permutation count of the main text. Fix even $n\ge2$, $m\ge1$, $N=nm$, $X=\ZZ_m\times\{0,\dots,n-1\}$ with blocks $B_j$, where $\eta,\Sigma\equiv\Sigma_B$ read each block forward/backward, $P(\tau)=\cyc(\tau)+\cyc(\eta^{-1}\tau)+\cyc(\Sigma^{-1}\tau)$, $\Cat_r=\binom{2r}{r}/(r+1)$, $D(n,m)=\#\{\tau:P(\tau)=\tfrac32N+2\}$.

\begin{lemma}[Refinement \cite{Biane1996,Biane1997}]\label{lem:refine}
If $|\alpha|+|\alpha^{-1}\beta|=|\beta|$, i.e. $\alpha$ is on a geodesic from the identity permutation to $\beta$, then every cycle of $\alpha$ lies in a cycle of $\beta$.
\end{lemma}

\begin{lemma}[Geodesic $\Leftrightarrow$ non-crossing \cite{Biane1997, NicaSpeicher2006}]\label{lem:geo}
For an $N$-cycle $\omega$ and a pairing $\mu$, $|\mu|+|\omega^{-1}\mu|=N-1$ if and only if $\mu$ is non-crossing in the cyclic order $\omega$.
\end{lemma}

\begin{lemma}[Dominant permutations are bi-noncrossing pairings]\label{lem:dom}
$P(\tau)=\tfrac32N+2$ if and only if $\tau$ is a perfect matching non-crossing in both the $\eta$ and $\Sigma$ orders.
\end{lemma}
\begin{proof}
($\Rightarrow$) Equality forces the three inequalities of Prop.~\ref{prop:pmax} to be equalities, which solve to $|\tau|=N/2$ and $|\eta^{-1}\tau|=|\Sigma^{-1}\tau| = N/2-1$. Lemma~\ref{lem:geo} then gives that $\tau$ is non-crossing in both orders. By Lemma~\ref{lem:refine}, cycles of $\eta^{-1}\tau$ lie in cycles of $\eta^{-1}\Sigma$ because $\eta^{-1}\Sigma$ has two cycles. These two cycles are the parity classes of $N$, so $\eta^{-1}\tau$ preserves the parity of $N$. i.e.~maps evens to evens and odds to odds. Since $\eta$ reverses the parity i.e.~maps all evens to odds and vice versa, $\tau$ must reverse it. In order to reverse the parity, all cycles of $\tau$ must have even length. With $N/2$ cycles, all cycles must be length $2$.  ($\Leftarrow$) A bi-noncrossing matching has $|\tau|=N/2$ and (Lemma~\ref{lem:geo}) $|\eta^{-1}\tau|=|\Sigma^{-1}\tau|=N/2-1$, so $P(\tau)=\tfrac32N+2$.
\end{proof}

\begin{lemma}[Block graph]\label{lem:block}
For bi-noncrossing $\tau$, the block graph (vertices are $B_i$ and edges are pairs that $\tau$ connects) is a non-crossing partial (meaning not every $B_i$ must be paired) pairing on the cyclic blocks.
\end{lemma}
\begin{proof}
If $B_i$ joins both $B_j$ and $B_\ell$ ($B_i<B_j<B_\ell$) via endpoints $x_j,x_\ell\in B_i$, non-crossing in the $\eta$ order forces $x_\ell<_{B_i^+}x_j$, while the $\Sigma$ order (same block order, $B_i$ reversed) forces $x_\ell<_{B_i^-}x_j$, i.e.\ $x_j<_{B_i^+}x_\ell$, a contradiction. If $B_i\!-\!B_k$ and $B_j\!-\!B_\ell$ with $i<j<k<\ell$, this would be crossing in the $\eta$ order. Thus, the connections between blocks must be non-crossing.
\end{proof}

\begin{lemma}[Local decorations]\label{lem:local}
An isolated block has $\Cat_{n/2}$ internal bi-noncrossing matchings, and a connected pair has $\Cat_n-\Cat_{n/2}^2$.
\end{lemma}
\begin{proof}
An isolated block has $\Cat_{n/2}$ non-crossing matchings on $n$ points \cite{Stanley1999}, preserved under reversal. For two blocks, a case check shows every matching non-crossing in the forward order $a_1<\dots<a_n<b_1<\dots<b_n$ stays non-crossing after reversing each block, so the two-block bi-noncrossing matchings are the $\Cat_n$ non-crossing matchings on $2n$ points, of which $\Cat_{n/2}^2$ are disconnected.
\end{proof}

\begin{lemma}[Fixed block graph]\label{lem:fixed}
For a non-crossing partial matching on $m$ blocks with $q$ edges, the number of bi-noncrossing $\tau$ is $\Cat_{n/2}^{\,m-2q}(\Cat_n-\Cat_{n/2}^2)^q$.
\end{lemma}
\begin{proof}
The $m-2q$ isolated blocks and $q$ pairs decorate independently, since distinct components of $G$ are separated or nested in the cyclic order and never alternate in both the $\eta$ and $\Sigma$ orders.
\end{proof}

\begin{theorem}[Dominant-permutation count]\label{thm:motzkin}
For all even $n\ge2$, $m\ge1$,
\begin{equation}
  D(n,m)=M_m\big(\Cat_{n/2},\,\Cat_n-\Cat_{n/2}^2\big)
  =\sum_{q=0}^{\lfloor m/2\rfloor}\binom{m}{2q}\Cat_q\,\Cat_{n/2}^{\,m-2q}\big(\Cat_n-\Cat_{n/2}^2\big)^q .
\end{equation}
\end{theorem}
\begin{proof}
By Lemma~\ref{lem:dom}, $D(n,m)$ counts bi-noncrossing matchings. By Lemma~\ref{lem:block} each has a non-crossing partial-matching block graph, of which there are $\binom{m}{2q}\Cat_q$ with $q$ edges \cite{Stanley1999}. Lemma~\ref{lem:fixed} gives the local count above each, and summing over $q$ yields $M_m$.
\end{proof}

\section{The spectral edge}\label{sec:edge}

The left edge of the saturation semicircle is $D\,\xi_{\min}^{\rm bulk}=\Chalf-2\sqrt{1-\Chalf^2}\approx-0.209<0$. Hypothesis (B) is that the expected minimum eigenvalue converges to it, the fluctuations being Tracy--Widom (GUE) \cite{TracyWidom1994} of order $D^{-1}(\dA\dB)^{-2/3}$, negligible against the $O(1/D)$ edge. Granting (A), (B), the mean terms cancel,
\begin{equation}
  \langle Z\rangle=\Big[\Chalf+\big(2\sqrt{1-\Chalf^2}-\Chalf\big)\Big]\sqrt{\tfrac{\dA\dB}{\dC}}=2\sqrt{1-\tfrac{64}{9\pi^2}}\sqrt{\tfrac{\dA\dB}{\dC}}\approx1.057\sqrt{\tfrac{\dA\dB}{\dC}}\ >\ \|\rhoTB\|_1\approx0.849\sqrt{\tfrac{\dA\dB}{\dC}},
\end{equation}
and the gap is $\log Z-\mathcal{E}=\log(2\sqrt{1-\Chalf^2}/\Chalf)\approx0.22$.

Concentration of measure pins $\xi_{\min}(\bineg)$ to a deterministic value at the $1/D$ scale, reducing (B$'$) to a bound on that value.

\begin{remark}[Self-averaging]\label{rem:selfavg}
The operational statements concern a typical state, so the macroscopic quantities must concentrate about their ensemble means, which they do. The trace norm $\|\rhoTB\|_1=\Tr|\rhoTB|$ is a linear spectral statistic, and for these Wishart-type ensembles such statistics have $O(1)$ variance by eigenvalue rigidity, so against the mean $\sim\dA\dB/D$ the relative fluctuation is $\sim1/(\dA\dB)$ and $\mathcal{E}=\log\|\rhoTB\|_1$ concentrates with $\mathrm{Var}(\mathcal{E})^{1/2}\sim1/(\dA\dB)$. The edge $D\,\xi_{\min}(\bineg)$ likewise concentrates, at the bulk value $\Chalf-2\sqrt{1-\Chalf^2}$ with Tracy--Widom fluctuations of order $(\dA\dB)^{-2/3}$ and no negative outliers. On the diagonal $\dA=\dB=\dC=2^N$ ($N=3,4,5$, up to $3000$ samples) we measure $\mathrm{Var}(\mathcal{E})^{1/2}\approx0.75/(\dA\dB)$ and $\mathrm{Var}(D\,\xi_{\min})^{1/2}\approx0.44\,(\dA\dB)^{-2/3}$, both vanishing, while the gap $\log Z-\mathcal{E}$ remains $O(1)$ for every sampled state (maximum $\lesssim0.24$). Hence $E_{ppt}/\mathcal{E}\to1$ for the typical state, not merely in the ensemble average.
\end{remark}

\end{document}